\documentclass[conference]{IEEEtran}
\usepackage{cite}
\usepackage{graphicx}
\usepackage{amssymb}
\usepackage{amsfonts}
\usepackage[cmex10]{amsmath}
\usepackage{array}
\usepackage{mdwmath}

\newtheorem{theorem}{\bf Theorem}

\newcommand{\xX}{\mathbf{X}}

\newcommand{\xW}{\mathbf{W}}
\newcommand{\xx}{\xX}
\newcommand{\Uu}{\mathcal{S}_I}
\newcommand{\Ss}{\mathcal{S}_J}
\providecommand{\Ii}{{\mathcal I}} 
\providecommand{\Jj}{\mathcal{J}}

\newcommand{\barxH}{\mathbf{\bar{H}}}
\newcommand{\barxX}{\mathbf{\bar{X}}}
\newcommand{\barxY}{\mathbf{\bar{Y}}}
\newcommand{\barxZ}{\mathbf{\bar{Z}}}
\IEEEoverridecommandlockouts

\begin{document}

\title{On the Secure Degrees of Freedom of Wireless $X$ Networks}

\author{\authorblockN{Tiangao Gou, Syed A. Jafar}
\authorblockA{Electrical Engineering and Computer Science\\
University of California Irvine, Irvine, California, 92697\\
Email: \{tgou,syed\}@uci.edu}
\thanks{The work of S. Jafar was
supported by ONR Young Investigator Award N00014-08-1-0872.}}
\maketitle

\begin{abstract}
Previous work showed that the $X$ network with $M$ transmitters, $N$
receivers has $\frac{MN}{M+N-1}$ degrees of freedom. In this work we
study the degrees of freedom of the $X$ network with secrecy
constraints, i.e. the $X$ network where some/all messages are
confidential. We consider the $M \times N$ network where all
messages are secured and show that $\frac{N(M-1)}{M+N-1}$ degrees of
freedom can be achieved. Secondly, we show that if messages from
only $M-1$ transmitters are confidential, then $\frac{MN}{M+N-1}$
degrees of freedom can be achieved meaning that there is no loss of
degrees of freedom because of secrecy constraints. We also consider
the achievable secure degrees of freedom under a more conservative
secrecy constraint. We require that messages from any subset of
transmitters are secure even if other transmitters are compromised,
i.e., messages from the compromised transmitter are revealed to the
unintended receivers. We also study the achievable secure degrees of
freedom of the $K$ user Gaussian interference channel under two
different secrecy constraints where $\frac{1}{2}$ secure degrees of
freedom per message can be achieved. The achievable scheme in all
cases is based on random binning combined with interference
alignment.
\end{abstract}

\section{Introduction}
Security is an important issue if the transmitted information is
confidential. Researchers have studied the information theoretic
secrecy for different channel models. In \cite{Wyner75}, Wyner first
proposed the wiretap channel model to characterize single user
secure communication problem, i.e., a sender transmits a
confidential message to its receiver while keeping a wire-tapper
totally ignorant of the message. The secrecy level is measured by
the equivocation rate, i.e., the entropy rate of the confidential
message conditioned on the received signal at the wire-tapper. More
recent information-theoretic research on secure communication
focuses on multi-user scenarios. In \cite{Yingbin:compound}, the
authors study the compound wire-tap channel where the sender
multicasts its messages to multiple receivers while ensuring the
confidentiality of the messages at multiple wire-tappers. Multiple
access channel with confidential messages has been studied in
\cite{Liang06isit:mac,Liu:ISIT:mac,Ender:GMACwiretap}. Broadcast
channel with confidential messages has been studied in
\cite{Csiszar:broadcast},\cite{Ruoheng:dmib}. The two user discrete
memoryless interference channel with confidential messages is
studied in \cite{Ruoheng:dmib}.

It is well known that the secrecy capacity of the Gaussian wiretap
channel is the difference between the capacities of the main and the
wiretap channels \cite{cheong}. In other words, there is a rate
penalty for ensuring the secrecy. From the degrees of freedom
perspective, this is pessimistic since the channel loses all its
degrees of freedom. Further, even the two user Gaussian interference
channel loses all its degrees of freedom if we have to ensure that
messages from both transmitters are confidential, i.e., a message
should remain secure from the undesired receiver. The results of the
Gaussian wiretap channel and the $2$ user Gaussian interference
channel prompt one to ask whether it is possible for a network to
have positive number of degrees of freedom if the messages in the
network are secure. The answer to this question lies in the study of
the $K$ user Gaussian interference channel with secure messages
\cite{onur} which indeed has positive number of degrees of freedom
if $K > 2$. It is shown that the network has $\frac{K(K-2)}{2K-2}$
secure degrees of freedom. The key to increase the secure degrees of
freedom is interference alignment. Interference signals associated
with the messages needed to be secured are aligned to occupy smaller
dimension so that the secrecy penalty rate is minimized. At the same
time, the degrees of the freedom for the legitimate channel is
maximized by interference alignment. Thus, the tool of interference
alignment serves the dual purpose of minimizing the secrecy penalty
rate and maximizing the rate of the legitimate messages, thus
improving the secure degrees of freedom of the network.

In this paper, we generalize the result of \cite{onur} to the $X$
network. We study the achievable secure degrees of freedom of the $M
\times N$ user wireless $X$ network, i.e., a network with $M$
transmitters and $N$ receivers where independent confidential
messages need to be conveyed from each transmitter to each receiver.
$X$ networks are interesting since they encompasses different
communication scenarios. For example, each transmitter is associated
with a broadcast channel, each receiver is associated with a
multiple access channel and every pair of transmitters and receivers
comprises an interference channel. In other words, broadcast
channel, multiple access channel and interference channel are
special cases of $X$ networks. In addition, interference alignment
is also feasible on $X$ networks. In \cite{Cadambe_Jafar_X},
interference alignment schemes are constructed to achieve
$\frac{1}{M+N-1}$ degrees of freedom per frequency/time slot for
each message without secrecy constraint. In this paper, we exploit
alignment of interference to assist secrecy in the network. We study
the achievable secure degrees of freedom under four different
secrecy constraints. We show that if the set of all unintended
messages is secured at each receiver, then each message can achieve
$\frac{M-1}{M}\frac{1}{M+N-1}$ secure degrees of freedom for a total
of $\frac{N(M-1)}{M+N-1}$ secure degrees of freedom. In other words,
only a fraction $\frac{1}{M}$ degrees of freedom is lost under this
secrecy constraint. Interestingly, if we only secure the set of
unintended messages from any $M-1$ transmitters at each receiver,
then each message can achieve $\frac{1}{M+N-1}$ secure degrees of
freedom which is the same as what one can achieve without secrecy
constraint. This corresponds to a scenario where one transmitter's
messages need not be secure, perhaps because their confidentiality
is ensured cryptographically, by some higher layer. In this case,
the other messages increase their degrees of freedom by exploiting
this. Next, we consider a more conservative secrecy constraint.
Transmitters do not trust each other, so we require that even if any
subset of transmitters $\mathcal{S}$ is compromised, i.e., the
messages from the compromised transmitter are revealed to the
unintended receivers (through a genie), the remaining transmitters'
messages are still secure. For this case, we show that if the set of
all unintended messages is secured then
$\frac{N(M-|\mathcal{S}|-1)}{M+N-1}$ secure degrees of freedom can
be achieved for the remaining $(M-|\mathcal{S}|)\times N$ users. If
we only need to secure the set of unintended messages from
$M-|\mathcal{S}|-1$ transmitters, then $\frac{1}{M+N-1}$ secure
degrees of freedom can be achieved for each message. The achievable
scheme for all cases is based on random binning combined with
interference alignment.

\section{System Model and Secrecy Constraints}
\subsection{System Model}
The $M \times N$ user $X$ network is comprised of $M$ transmitters
and $N$ receivers. Each transmitter has an independent message for
each receiver. The channel output at the $j^{th}$ receiver over the
$f^{th}$ frequency slot and the $t^{th}$ time slot is described as
follows:
\begin{equation*}
Y_j(f,t)=\sum_{i=1}^{M}H_{ji}(f)X_i(f,t)+Z_j(f,t),~j=1,2,\ldots,N
\end{equation*}
where $X_i(f,t)$ is the input signal at Transmitter $i$, $H_{ji}(f)$
is the channel coefficient from Transmitter $i$ to Receiver $j$ and
$Z_j(f,t)$ represents the additive white Gaussian noise (AWGN) at
Receiver $j$. We assume the channel coefficients vary across
frequency slots but remain constant in time and are drawn from a
continuous distribution. We assume all channel coefficients are
known to all transmitters and receivers. Using the symbol extension
channel in \cite{Cadambe_Jafar_X}, the input-output relationship is
characterized as follows:
\begin{equation}
\barxY_j(t) = \displaystyle\sum_{i=1}^{M}\barxH_{ji}
\barxX_i(t)+\barxZ_j(t)
\end{equation}
where $\barxX_i(t)$ is the $F \times 1$ column vector representing
the $F$ symbol extension of the transmitted symbol $X_i$, i.e.,
$\barxX_i(t)= [ X_i(1,t)~ X_i(2,t)~ \cdots ~X_i(F,t)]^T$. Similarly,
$\barxY_j(t)$ and $\barxZ_j(t)$ represent the symbol extension of
$Y_j$ and $Z_j$, respectively. $\barxH_{ji}$ is the $F \times F$
diagonal matrix representing the extension of the channel, i.e.,
\begin{eqnarray*}
\barxH_{ji}= \left[\begin{array}{cccc}H_{ji}(1)& 0 & \cdots & 0\\
0 & H_{ji}(2) & \cdots & 0\\ \vdots & \vdots & \ddots & \vdots\\ 0 &
0 & \cdots & H_{ji}(F) \end{array}\right]
\end{eqnarray*}
Transmitter $i$ has message $W_{ji} \in \{1,2,\ldots,M_{ji}\}$ for
Receiver $j$, for each $i \in \{1,2,\ldots,M\}$, $j \in
\{1,2,\ldots,N\}$, resulting in a total of $MN$ independent
messages. An $(M_{11},\ldots,M_{NM},n,F,P_e)$ code
for the $X$ channel consists of the following:\\
\begin{itemize}
\item $MN$ independent message sets:
$\mathcal{W}_{ji}=\{1,2,\ldots,M_{ji}\}$
\item $M$ encoding functions,  $f_i$:
$\mathcal{W}_{1i} \times \mathcal{W}_{2i} \times \cdots \times
\mathcal{W}_{Ni} \to \barxX^n_i$, where $\barxX^n_i=[\barxX_i(1)~
\barxX_i(2)~ \cdots ~\barxX_i(n)]$, which map the message tuple
$(w_{1i}, w_{2i}, \cdots, w_{Ni}) \in \mathcal{W}_{1i} \times
\mathcal{W}_{2i} \times \cdots \times \mathcal{W}_{Ni}$ to
transmitted symbols. Each transmitter has a power constraint, i.e.
\begin{equation*}
\frac{1}{nF}\sum_{f=1}^F\sum_{t=1}^n |X_i(t,f)|^2 \leq P, ~i\in
\{1,2,\ldots,M\}.
\end{equation*}

\item $N$ decoding functions, $g_j$: $\barxY_j^n \to \mathcal{W}_{j1} \times
\mathcal{W}_{j2} \times \cdots \times \mathcal{W}_{jM}$, where
$\barxY_j^n=[\barxY_j(1)~\barxY_j(2)~ \cdots ~\barxY_j(n)]$, which
map the received sequence $\barxY_j^n$ to the decoded message tuple
$(\hat{w}_{j1}, \hat{w}_{j2}, \cdots ,\hat{w}_{jM}) \in
{\mathcal{W}_{j1} \times \mathcal{W}_{j2} \times \cdots \times
\mathcal{W}_{jM}}$.
\end{itemize}
The maximal average probability of error $P_e$ for an
$(M_{11},\ldots, M_{NM}, n, F, P_e)$ code is defined as
\begin{equation*}
P_e \triangleq \text{max}\{P_{e,11},P_{e,21},\cdots,P_{e,NM}\}
\end{equation*}
where
\begin{equation*}
P_{e,ji}=\frac{1}{M_{ji}}\sum_{w_{ji} \in
\mathcal{W}_{ji}}P\{g(\barxY_j^n)\neq w_{ji}|w_{ji} ~\text{sent}\}
\end{equation*}
We use the equivocation rate $\frac{1}{nF}H(W|\barxY^n_j)$ as the
secrecy measure.

A rate tuple $(R_{11}, R_{21}, \ldots, R_{NM})$ is said to be
achievable for the $M \times N$ user $X$ network with confidential
messages if for any $\epsilon > 0$, there exists an
$(M_{11},\ldots,M_{NM},n,F,P_e)$ code such that
\begin{equation*}
\frac{1}{nF}\log_2(M_{ji}) \geq R_{ji}
\end{equation*}
and the reliability requirement
\begin{equation*}
P_e \leq \epsilon
\end{equation*}
and the security constraints which will be defined shortly are
satisfied. The secure degrees of freedom tuple
$(\eta_{11},\ldots,\eta_{NM})$ is achievable if the rate tuple
$(R_{11},\ldots,R_{NM})$ is achievable and
\begin{eqnarray*}
&&\eta_{ji}=\lim_{P \rightarrow \infty} \frac{R_{ji}(P)}{\log(P)} ~~
 \forall (j,i)\in \mathcal{J} \times
\mathcal{I}, \\&&
\mathcal{I}=\{1,2,\ldots,M\},~\mathcal{J}=\{1,2,\ldots,N\}
\end{eqnarray*}

\subsection{Secrecy Constraints}
We will define four different secrecy constraints as follows:
\subsubsection{Secrecy Constraint 1}
The secrecy constraint is defined as
\begin{equation*}
\frac{1}{nF}H(\mathbf{W}_{(\mathcal{J}-j) \times
\mathcal{I}}|\barxY^n_j) \geq \sum_{(r,i)\in (\mathcal{J}-j) \times
\mathcal{I}}R_{ri}-\epsilon
\end{equation*}
where
\begin{equation*}
\mathbf{W}_{(\mathcal{J}-j) \times \mathcal{I}} \triangleq \{W_{ri}:
\forall (r,i) \in (\mathcal{J}-j) \times \mathcal{I} \}
\end{equation*}
This ensures the perfect secrecy of the set of all unintended
messages at each receiver. Moreover, it can be shown that perfect
secrecy for a set of messages guarantees perfect secrecy for any
subset of that message set, i.e.,
\begin{equation}
\frac{1}{nF}H(\mathbf{W}_S|\barxY^n_j) \geq \sum_{(r,i)\in
S}R_{ri}-\epsilon~~\forall S \subseteq (\mathcal{J}-j) \times
\mathcal{I} \label{collreq}
\end{equation}
To see this, consider
\begin{eqnarray}
H(\mathbf{W}_{(\mathcal{J}-j) \times \mathcal{I}}|\barxY^n_j) &=&
H(\mathbf{W}_S|\barxY^n_j)+ H(\mathbf{W}_{S^c}|\mathbf{W}_S,\barxY^n_j)\notag\\ \label{chainr}\\
&\leq&
H(\mathbf{W}_S|\barxY^n_j)+H(\mathbf{W}_{S^c})\label{condition}
\end{eqnarray}
where $S^c$ denotes the complement of $S$ and \eqref{chainr} follows
from the chain rule and \eqref{condition} follows from the fact that
conditioning reduces the entropy. If the message set satisfies the
secrecy constraint, i.e.,
\begin{equation*}
H(\mathbf{W}_{(\mathcal{J}-j) \times \mathcal{I}}|\barxY^n_j)\geq
H(\mathbf{W}_S)+H(\mathbf{W}_{S^c})-\epsilon
\end{equation*}
then from \eqref{condition} we have
\begin{eqnarray*}
&H(\mathbf{W}_S|\barxY^n_j)+H(\mathbf{W}_{S^c}) \geq H(\mathbf{W}_S)+H(\mathbf{W}_{S^c})-\epsilon\\
&\Rightarrow H(\mathbf{W}_S|\barxY^n_j)\geq H(\mathbf{W}_S)-\epsilon
\end{eqnarray*}
Thus, the confidentiality of the subset $\mathbf{W}_S$ is preserved.
\subsubsection{Secrecy Constraint 2}
Instead of ensuring the confidentiality of the set of unintended
messages of all transmitters, we only secure the set of unintended
messages from any $M-1$ transmitters. Secrecy constraint 2 is
defined as
\begin{equation*}
\frac{1}{nF}H(\mathbf{W}_{(\Jj-j) \times (\Ii-l)}|\barxY^n_j) \geq
\sum_{(r,i)\in (\mathcal{J}-j) \times
(\Ii-l)}R_{ri}-\epsilon~\forall l \in \Ii
\end{equation*}
where
\begin{equation*}
\mathbf{W}_{(\Jj-j) \times (\Ii-l)} = \{W_{ri}:\forall (r,i) \in
(\Jj-j) \times (\Ii-l)\}
\end{equation*}
Again, the perfect secrecy of a message set guarantees perfect
secrecy for any subset of that message set, i.e.,
\begin{eqnarray*}
\frac{1}{nF}H(\mathbf{W}_{\Ss \times \Uu}|\barxY^n_j) \geq
\sum_{(r,i)\in \mathcal{S}_J \times
\mathcal{S}_I}R_{ri}-\epsilon,\\\forall \mathcal{S}_J \subseteq
\mathcal{J}-j, \forall \mathcal{S}_I \subseteq \mathcal{I}-l
\end{eqnarray*}
Note that satisfying secrecy constraint 1 ensures satisfying secrecy
constraint 2.
\subsubsection{Secrecy Constraint 3}
Let us define $\mathcal{S}_I \subset \mathcal{I}$ to be the set of
transmitters that are compromised, i.e., the messages from the
compromised transmitter are revealed to the unintended receivers and
$\mathcal{S}^c_I$ to be the set of the remaining transmitters. We
define secrecy constraint 3 as
\begin{eqnarray*}
\frac{1}{nF}H(\mathbf{W}_{(\Jj-j) \times
\Uu^c}|\barxY^n_j,\xW_{(\mathcal{J}-j) \times \mathcal{S}_I}) \geq
\sum_{(r,i)\in (\mathcal{J}-j) \times
\mathcal{S}_I^c}R_{ri}-\epsilon \\ \forall \mathcal{S}_I \subset
\mathcal{I}
\end{eqnarray*}
This constraint ensures that secrecy of any subset of transmitters
even if all other transmitters are compromised. Also, this secrecy
constraint guarantees that
\begin{eqnarray*}
\frac{1}{nF}H(\mathbf{W}_{\Ss \times
\Uu^c}|\barxY^n_j,\xW_{(\mathcal{J}-j) \times \mathcal{S}_I}) \geq
\sum_{(r,i)\in \mathcal{S}_J \times
\mathcal{S}_I^c}R_{ri}-\epsilon\\
\forall \mathcal{S}_I \subset \mathcal{I},\forall \mathcal{S}_J
\subseteq \mathcal{J}-j
\end{eqnarray*}

\subsubsection{Secrecy Constraint 4}
Even if any subset of transmitters $\mathcal{S}_I \subset
\mathcal{I}$ is compromised, we require secrecy of the set of
messages from $\mathcal{S}_I^c-l$ transmitters for any $l \in
\mathcal{S}_I^c$. We define secrecy constraint 4 as
\begin{eqnarray*}
&&\frac{1}{nF}H(\mathbf{W}_{(\Jj-j) \times
(\Uu^c-l)}|\barxY^n_j,\xW_{(\mathcal{J}-j) \times \mathcal{S}_I})
\\ &\geq& \sum_{(r,i)\in (\mathcal{J}-j) \times
(\Uu^c-l)}R_{ri}-\epsilon ~~~~\forall \mathcal{S}_I \subset
\mathcal{I},~\forall l \in \Uu^c
\end{eqnarray*}

\section{The $M \times N$ user $X$ Network with confidential messages}
In this section, we consider the achievable secure degrees of
freedom of the $M \times N$ user $X$ channel under different secrecy
constraints. In order to satisfy the secrecy constraints, we use the
random binning coding scheme to generate the codebook. This is a
natural extension of the coding scheme used in \cite{Ruoheng:dmib}
to achieve the inner bound of the capacity region of the two user
discrete memoryless interference channel with confidential messages.
To maximize the achievable degrees of freedom, we adopt the
interference alignment scheme used in \cite{Cadambe_Jafar_X}. The
main results of this section are presented in the following
theorems:
\begin{theorem}
For the $M \times N$ user $X$ network with single antenna nodes,
$\frac{M-1}{M(M+N-1)}$ secure degrees of freedom can be achieved for
each message $W_{ji}$, $\forall j \in \{1,\ldots, N\}, \forall i \in
\{1,\ldots, M\}$ and hence a total of $\frac{N(M-1)}{M+N-1}$ secure
degrees of freedom can be achieved under secrecy constraint
1.\label{thm:col}
\end{theorem}
\begin{proof}
We provide a detailed proof in the Appendix. A sketch of the proof
is provided here. Consider the $F$ symbol extension channel where
$F=N(m+1)^{\Gamma}+(M-1)m^{\Gamma}, \forall m \in \mathbb{N}$ and
$\Gamma =(N-1)(M-1)$. Over the $F$ symbol extension channel, message
$W_{j1}$ is encoded at Transmitter 1 into $m_1=(m+1)^\Gamma$
independent streams $\xX_{j1}(t)$ which is an $(m+1)^\Gamma \times
1$ vector and message $W_{ji}, i \neq 1$ is encoded at Transmitter
$i$ into $m_i=m^\Gamma$ independent streams $\xX_{ji}(t)$ which is
an $m^\Gamma \times 1$ vector based on random binning coding scheme.
Note that such coding scheme introduces randomness to ensure the
secrecy. Then transmitter $i$ employs the interference alignment
scheme mapping $\xX_{ji}(t)$ to $\mathbf{V}_{ji}(t)\xX_{ji}(t)$
where $\mathbf{V}_{ji}$ is the $F \times m_i$ matrix. At last,
Transmitter $i$ sends signal
$\barxX_i(t)=\sum_{j=1}^N\mathbf{V}_{ji}(t)\xX_{ji}(t)$ into the
channel. Note that the precoding matrices $\mathbf{V}_{ji}(t)$ are
chosen as given in \cite{Cadambe_Jafar_X} so that at each receiver,
the desired signal vectors span a signal space which is disjoint
with the space spanned by the interference vectors. Therefore, each
receiver can decode its desired data streams by zero forcing the
interference. Note that at Receiver $j$, the signal vectors
associated with $M$ desired messages $W_{ji}, \forall i=1,\ldots,M$
span a $(m+1)^{\Gamma}+(M-1)m^{\Gamma}$ dimensional subspace in the
$F=N(m+1)^{\Gamma}+(M-1)m^{\Gamma}$ dimensional signal space. Thus,
to get an interference-free signal subspace, the dimension of the
subspace spanned by all interference vectors has to be less than or
equal to $(N-1)(m+1)^{\Gamma}$. Notice that the interference vectors
from Transmitter 1 span a $(N-1)(m+1)^{\Gamma}$ dimensional
subspace. Therefore, we can align the interference vectors from all
other transmitters within this subspace so that each receiver can
decode its desired data streams by zero forcing the interference in
this subspace. Next, it can be shown that the following secrecy rate
is achievable:
\begin{eqnarray}
R_{ji}&=&\frac{1}{F}I(\xX_{ji};\barxY_j) \notag
\\&-&\frac{1}{F}\frac{1}{M(N-1)}\max_{k \in \Jj} I(\xX_{(\Jj-k)\times
\Ii};\barxY_k|\xX_{k\times \Ii})\notag\\
&&~~~~~~~~~~~~~~~~~~~~~~~~~~~~~~~~~\forall (j,i) \in \Jj \times \Ii
\label{achieve1}
\end{eqnarray}
From \cite{Cadambe_Jafar_X}, we have
\begin{equation*}
I(\xX_{ji};\barxY_j)=(m+1)^{\Gamma}\log(P)+o(\log(P))~~i=1
\end{equation*}
and
\begin{equation*}
I(\xX_{ji};\barxY_j)= m^{\Gamma}\log(P)+o(\log(P))~~i=2,\ldots,M
\end{equation*}
Next, consider the term $I(\xX_{(\Jj-k)\times
\Ii};\barxY_k|\xX_{k\times \Ii})$ which denotes the secrecy penalty.
Notice that all the interference vectors are aligned within the
space spanned by $(N-1)(m+1)^{\Gamma}$ interference vectors from
Transmitter 1. Therefore, the secrecy penalty is
\begin{eqnarray*}
&~&I(\xX_{(\Jj-k)\times \Ii};\barxY_k|\xX_{k\times \Ii})\\
&=&(N-1)(m+1)^{\Gamma}\log(P)+o(\log(P))~~ \forall k \in \Jj
\end{eqnarray*}
Hence, \eqref{achieve1} can be written as
\begin{equation*}
R_{ji}=
\frac{1}{F}(m+1)^{\Gamma}(1-\frac{1}{M})\log(P)+o(\log(P))~~i=1
\end{equation*}
and
\begin{equation*}
R_{ji}=
\frac{1}{F}(m^{\Gamma}-\frac{(m+1)^{\Gamma}}{M})\log(P)+o(\log(P))~~i=2,\ldots,M
\end{equation*}
As $m \to \infty$, we have
\begin{equation*}
R_{ji}=\frac{M-1}{M(M+N-1)}\log(P)+o(\log(P))~~~ \forall (j,i)\in
\Jj \times \Ii
\end{equation*}
As a result, each message can achieve
$\eta_{ji}=\frac{M-1}{M(M+N-1)}$ secure degrees of freedom for a
total of $\frac{(M-1)N}{M+N-1}$ secure degrees of freedom.
\end{proof}
Note that in \cite{Cadambe_Jafar_X}, it is shown that
$\frac{1}{M+N+1}$ degrees of freedom can be achieved for each
message $W_{ji}$ without secrecy constraint. Theorem $\ref{thm:col}$
shows that only a fraction $\frac{1}{M}$ degrees of freedom is lost
under secrecy constraint 1. However, it is interesting that if we
relax the secrecy constraint a little, i.e., only ensure the
confidentiality of the set of messages from any $M-1$ out of $M$
transmitters at each receiver, there will be no loss of degrees of
freedom. We present the result in the following theorem:
\begin{theorem}
For the $M \times N$ user $X$ network with single antenna nodes,
each message can achieve $\frac{1}{M+N-1}$ secure degrees of freedom
for a total of $\frac{MN}{M+N-1}$ secure degrees of freedom under
secrecy constraint 2.
\end{theorem}
\begin{proof}
The proof is similar to the proof of Theorem 1. We only provide a
sketch of proof here. It can be shown that the following secrecy
rate is achievable:
\begin{eqnarray}
R_{ji} =
\frac{1}{F}I(\xX_{ji};\barxY_j)~~~~~~~~~~~~~~~~~~~~~~~~~~~~~~~~~~~~~~~~~~~~~\notag\\-
\frac{1}{F}\frac{1}{(M-1)(N-1)} \max_{k \in \Jj, l \in
\Ii}I(\xX_{(\Jj-k)\times(\Ii-l)};\barxY_k|\xX_{k \times \Ii}) \notag\\
\forall (j,i) \in \Jj \times \Ii ~\forall l \in \Ii \label{achinv}
\end{eqnarray}
where $F=N(m+1)^{\Gamma}+(M-1)m^{\Gamma}$ and $\Gamma=(M-1)(N-1)$.
Through interference alignment, it can be shown that
\begin{equation*}
I(\xX_{ji};\barxY_j)=\eta\log(P)+o(\log(P))
\end{equation*}
where $\eta=(m+1)^{\Gamma}$ when $i=1$ and $\eta=m^{\Gamma}$ when
$i=2,3,\ldots,M$. Then consider the secrecy penalty term
$I(\xX_{(\Jj-k)\times(\Ii-l)};\barxY_k|\xX_{k \times \Ii})$. At each
receiver, the interference vectors from Transmitter $2,3\ldots,M$
are aligned perfectly with the interference vectors from Transmitter
1, i.e. every interference signal vector from Transmitter
$2,3\ldots,M$ is aligned along the same dimension with one
interference signal vector from Transmitter 1. Note that there are
$(m+1)^{\Gamma}$ interference vectors for each message from
Transmitter 1, but there are only $m^{\Gamma}$ interference vectors
for each message from Transmitter $2,3\ldots,M$. If $l=1$,
$I(\xX_{(\Jj-k)\times(\Ii-1)};\barxY_k|\xX_{k \times \Ii})$ denotes
the mutual information between the channel output at Receiver $k$
and channel inputs from Transmitter $2,\ldots,M$. Since all vectors
from Transmitter $2,3\ldots,M$ are aligned perfectly with
interference vectors from Transmitter 1, it has zero degrees of
freedom, i.e., $I(\xX_{(\Jj-k)\times(\Ii-1)};\barxY_k|\xX_{k \times
\Ii})=o(\log(P))$. For $\forall l \neq 1$, the interference vectors
from Transmitter $l$ occupy a $(N-1)m^{\Gamma}$ dimensional
subspace. Therefore, the remaining transmitters can get a
$(N-1)((m+1)^{\Gamma}-m^{\Gamma})$ dimensional space without
interference vectors from Transmitter $l$. Therefore, we have
\begin{eqnarray*}
\max_{k \in \Jj,~l \in
\Ii}I(\xX_{(\Jj-k)\times(\Ii-l)};\barxY_k|\xX_{k \times
\Ii})~~~~~~~~~~\\=(N-1)((m+1)^{\Gamma}-m^{\Gamma})\log(P)+o(\log(P))\\ \forall (j,i) \in \Jj \times \Ii\\
\end{eqnarray*}
Thus, \eqref{achinv} can be written as
\begin{eqnarray*}
R_{ji} =
\frac{(M-1)\eta-((m+1)^{\Gamma}-m^{\Gamma})}{F(M-1)}\log(P)+o(\log(P))\\
\forall i=1,2,\ldots,M
\end{eqnarray*}
When $m \to \infty$, we have
\begin{equation*}
\eta_{ji}=\lim_{m \to
\infty}\frac{(M-1)\eta-((m+1)^{\Gamma}-m^{\Gamma})}{F(M-1)}=\frac{1}{M+N-1}
\end{equation*}
Therefore, each message can achieve $\frac{1}{M+N-1}$ secure degrees
of freedom for a total of $\frac{MN}{M+N-1}$ secure degrees of
freedom.
\end{proof}
Next, we consider the achievable secure degrees of freedom under the
more conservative secrecy constraints to ensure secrecy of any
subset of transmitters even if all other transmitters are
compromised. We present the result in the following theorem.
\begin{theorem}
For the $M \times N$ user $X$ network with single antenna nodes,
even if any subset of transmitters, $\mathcal{S} \subset
\{1,\ldots,M\}$ is compromised, the remaining $(M-|\mathcal{S}|)
\times N$ users can still achieve a total of
$\frac{N(M-|\mathcal{S}|-1)}{M+N-1}$ secure degrees of freedom under
secrecy constraint 3 and $\frac{N(M-|\mathcal{S}|)}{M+N-1}$ secure
degrees of freedom under secrecy constraint 4, as long as
$|\mathcal{S}|\leq M-2$ .
\end{theorem}
\begin{proof}
To satisfy secrecy constraint 3, we design an achievable scheme to
satisfy the following secrecy constraint:
\begin{eqnarray*}
\frac{1}{nF}H(\mathbf{W}_{(\Jj-j) \times
\mathcal{S}^c}|\barxY^n_j,\xX^n_{(\mathcal{J}-j) \times
\mathcal{S}}) \geq \sum_{(r,i)\in (\mathcal{J}-j) \times
\mathcal{S}^c}R_{ri}-\epsilon
\\ \forall \mathcal{S} \subset \mathcal{I}
\end{eqnarray*}
where
\begin{equation*}
\xX^n_{(\mathcal{J}-j) \times \mathcal{S}}= \{\xX^n_{ji}:\forall
(j,i) \in (\Jj-j) \times \mathcal{S}\}
\end{equation*}
$\xX^n_{ji}$ denotes the codeword for message $W_{ji}$. Note that
this secrecy constraint is stronger than
$\frac{1}{nF}H(\mathbf{W}_{(\Jj-j) \times
\mathcal{S}^c}|\barxY^n_j,\xW_{(\mathcal{J}-j) \times
\mathcal{S}})$. Because
\begin{eqnarray*}
&&H(\mathbf{W}_{(\Jj-j) \times
\mathcal{S}^c}|\barxY^n_j,\xW_{(\mathcal{J}-j) \times \mathcal{S}})\\
&\geq& H(\mathbf{W}_{(\Jj-j) \times
\mathcal{S}^c}|\barxY^n_j,\xW_{(\mathcal{J}-j) \times
\mathcal{S}},\xX^n_{(\mathcal{J}-j)
\times \mathcal{S}})\\
&=&H(\mathbf{W}_{(\Jj-j) \times
\mathcal{S}^c}|\barxY^n_j,\xX^n_{(\mathcal{J}-j) \times
\mathcal{S}})
\end{eqnarray*}
In other words, we want to ensure secrecy of any subset of
transmitters even if all other transmitters' codewords rather than
messages are revealed to the unintended receivers. This is possible
because the achievability scheme encodes the messages separately and
each message has its codewords. The coding scheme is similar to that
used in Theorem 1. Then it can be shown that the following secrecy
rate is achievable:
\begin{eqnarray*}
R_{ji}=\frac{1}{F}I(\xX_{ji};\barxY_j)-\frac{1}{F}\frac{1}{(M-|\mathcal{S}|)(N-1)}\times\\
\max_{k \in \Jj, \mathcal{S} \subset \Ii} I(\xX_{(\Jj-k)\times
\mathcal{S}^c};\barxY_k|\xX_{k \times
\mathcal{S}^c},\xX_{\mathcal{J} \times \mathcal{S}})\\\forall (j,i)
\in \Jj \times \mathcal{S}^c
\end{eqnarray*}
Consider the term $I(\xX_{(\Jj-k)\times
\mathcal{S}^c};\barxY_k|\xX_{k \times
\mathcal{S}^c},\xX_{\mathcal{J} \times \mathcal{S}})$. Following
similar analysis in Theorem 1, if $|\mathcal{S}|\leq M-2$, it can be
shown that
\begin{eqnarray*}
\max_{k \in \Jj, \mathcal{S} \subset \Ii} I(\xX_{(\Jj-k)\times
\mathcal{S}^c};\barxY_k|\xX_{k \times
\mathcal{S}^c},\xX_{\mathcal{J} \times
\mathcal{S}})\\=(N-1)(m+1)^{\Gamma}\log(P)+o(\log(P))
\end{eqnarray*}
Therefore,
\begin{eqnarray*}
R_{ji}=\frac{1}{F}(\eta-\frac{(m+1)^{\Gamma}}{M-|\mathcal{S}|})\log(P)+o(\log(P))\\
\forall (j,i) \in \Jj \times \mathcal{S}^c
\end{eqnarray*}
where $\eta=(m+1)^{\Gamma}$ when $i=1$ and $\eta=m^{\Gamma}$ when
$i=2,3,\ldots,M$. As $m \to \infty$,
\begin{eqnarray*}
R_{ji}=\frac{1}{M+N-1}(1-\frac{1}{M-|\mathcal{S}|})\log(P)+o(\log(P))\\
\forall (j,i) \in \Jj \times \mathcal{S}^c
\end{eqnarray*}
Therefore, each message can achieve
$\frac{1}{M+N-1}(1-\frac{1}{M-|\mathcal{S}|})$ secure degrees of
freedom for a total of $\frac{N(M-|\mathcal{S}|-1)}{M+N-1}$ secure
degrees of freedom under secrecy constraint 3.

Similarly, to satisfy secrecy constraint 4, we design an achievable
scheme to satisfy the following constraint:
\begin{eqnarray*}
&&\frac{1}{nF}H(\mathbf{W}_{(\Jj-j) \times
(\mathcal{S}^c-l)}|\barxY^n_j,\xX^n_{(\mathcal{J}-j) \times
\mathcal{S}})
\\ &\geq& \sum_{(r,i)\in (\mathcal{J}-j) \times
(\mathcal{S}^c-l)}R_{ri}-\epsilon ~~~~\forall \mathcal{S} \subset
\mathcal{I},~\forall l \in \Uu^c
\end{eqnarray*}
Then it can be shown that the following secure rate is achievable:
\begin{eqnarray*}
R_{ji} = \frac{1}{F}I(\xX_{ji};\barxY_j)-
\frac{1}{F}\frac{1}{(M-|\mathcal{S}|-1)(N-1)} \times \\ \max_{k \in
\Jj, l \in \mathcal{S}^c, \mathcal{S}\subset
\Ii}I(\xX_{(\Jj-k)\times(\mathcal{S}^c-l)};\barxY_k|\xX_{k \times
\mathcal{S}^c},\xX_{\mathcal{J} \times \mathcal{S}}) \\ \forall
(j,i) \in \Jj \times \mathcal{S}^c \notag
\end{eqnarray*}
Following similar analysis in Theorem 2, if $|\mathcal{S}|\leq M-2$,
it can be shown that
\begin{eqnarray*}
\max_{k \in \Jj, l \in \mathcal{S}^c, \mathcal{S}\subset
\Ii}I(\xX_{(\Jj-k)\times(\mathcal{S}^c-l)};\barxY_k|\xX_{k \times
\mathcal{S}^c},\xX_{\mathcal{J} \times
\mathcal{S}})\\=(N-1)((m+1)^{\Gamma}-m^{\Gamma})\log(P)+o(\log(P))
\end{eqnarray*}
Therefore,
\begin{eqnarray*}
R_{ji}=\frac{(M-|\mathcal{S}|-1)\eta-((m+1)^{\Gamma}-m^{\Gamma})}{F(M-|\mathcal{S}|-1)}\log(P)+o(\log(P))\\
\forall (j,i) \in \Jj \times \mathcal{S}^c
\end{eqnarray*}
where $\eta=(m+1)^{\Gamma}$ when $i=1$ and $\eta=m^{\Gamma}$ when
$i=2,3,\ldots,M$. As $m \to \infty$,
\begin{eqnarray*}
R_{ji}=\frac{1}{M+N-1}\log(P)+o(\log(P))~ \forall (j,i) \in \Jj
\times \mathcal{S}^c
\end{eqnarray*}
Therefore, each message can achieve $\frac{1}{M+N-1}$ secure degrees
of freedom for a total of $\frac{N(M-|\mathcal{S}|)}{M+N-1}$ secure
degrees of freedom.
\end{proof}

\section{The $K$ user Gaussian interference channel with confidential messages}
In this section, we consider the $K$ user Gaussian interference
channel with confidential messages. In \cite{onur}, this
interference channel with confidential messages is considered under
secrecy constraint 1, i.e.,
\begin{eqnarray*}
\frac{1}{nF}H(\mathbf{W}_{(\mathcal{K}-j)}|\barxY_j^n)\geq \sum_{i
\in (\mathcal{K}-j)}R_{i}-\epsilon \quad \forall j \in
\mathcal{K}=\{1,2,\ldots,K\}
\end{eqnarray*}
It is shown that each user can achieve $\frac{K-2}{2K-2}$ secure
degrees of freedom. However, we consider the same channel under
secrecy constraint 2, i.e.,
\begin{eqnarray*}
\frac{1}{nF}H(\mathbf{W}_{(\mathcal{K}-j-m)}|\barxY_j^n)\geq \sum_{i
\in (\mathcal{K}-j-m)}R_{i}-\epsilon \\ \forall j,m \in
\mathcal{K}=\{1,2,\ldots,K\}, j \neq m
\end{eqnarray*}
and secrecy constraint 4, i.e.,
\begin{eqnarray*}
\frac{1}{nF}H(\mathbf{W}_{(\mathcal{S}^c-j-m)}|\barxY^n_j,\xW_{\mathcal{S}})
\geq \sum_{i\in (\mathcal{S}^c-j-m)}R_{i}-\epsilon\\ \forall m \in
\mathcal{S}^c, ~ j \neq m,~\forall \mathcal{S} \subset
\mathcal{K}=\{1,2,\ldots,K\}
\end{eqnarray*}
where $\mathcal{S}$ is the set of users that are compromised.
Interestingly, we show that for these two scenarios, each message
can achieve $\frac{1}{2}$ secure degrees of freedom which is the
same as what one can achieve without secrecy constraint. We present
the results in the following theorems:
\begin{theorem}
For the $K$ user Gaussian interference channel with single antenna
nodes, each user can achieve $\frac{1}{2}$ secure degrees of freedom
for a total of $\frac{K}{2}$ secure degrees of freedom under secrecy
constraint 2.
\end{theorem}
\begin{proof}
The proof is similar to the proof of Theorem 2 and is omitted here.
\end{proof}
\begin{theorem}
For the $K$ user Gaussian interference channel with single antenna
nodes, even if any subset of users, $\mathcal{S} \subset
\{1,2,\ldots,K\}$ is compromised, then the remaining
$K-|\mathcal{S}|$ users can still achieve $\frac{1}{2}$ secure
degrees of freedom for each message for a total of
$\frac{K-|\mathcal{S}|}{2}$ secure degrees of freedom as long as
$|\mathcal{S}|<K-2$.
\end{theorem}
\begin{proof}
The proof is similar to the proof of Theorem 3 and is omitted here.
\end{proof}

\section{Conclusion}
In this work, we obtain the achievable secure degrees of freedom for
the $M \times N$ user $X$ network under different secrecy
constraints. We also obtain the achievable secure degrees of freedom
for the $K$ user Gaussian interference channel under two different
secrecy constraints. We can see another advantage of interference
alignment, i.e., interference signals are aligned along the same
dimensions to assist secrecy in wireless communications.

\section{Appendix}
\subsection{Proof of Theorem 1}
\label{app:proof}
\begin{proof}
Let $\Gamma =(N-1)(M-1)$ and $F=N(m+1)^{\Gamma}+(M-1)m^{\Gamma},
\forall m \in \mathbb{N}$. Over the $F$ symbol extension channel,
for each message $W_{ji}$, we generate
$2^{nF(R_{ji}+R^1_{ji}+R^2_{ji}+\cdots+R^{N-1}_{ji}+R^{\dagger}_{ji})}$
codewords each of length $nm_i$, where $m_1=(m+1)^{\Gamma}$,
$m_i=m^{\Gamma},\forall i=2,3,\ldots,M$ . Each element of the
codewords is i.i.d. $\sim \mathcal{CN}(0,\frac{P-\epsilon}{c})$ such
that the power constraint is satisfied. We denote the codeword as
\begin{equation*}
\xX^n(w_{ji},b^1_{ji},b^2_{ji},\ldots,b^{N-1}_{ji},b^{\dagger}_{ji})=[\xX_{ji}(1)~\cdots~\xX_{ji}(n)].
\end{equation*}
where $w_{ji} \in \{1,\ldots,2^{nFR_{ji}}\}$, $b^k_{ji} \in
\{1,\ldots,2^{nFR^k_{ji}}\},\forall k=1,\cdots,N-1$,
$b^{\dagger}_{ji} \in \{1,\ldots,2^{nFR^{\dagger}_{ji}}\}$ and
$\xX_{ji}(t)$ is an $m_i \times 1$ vector. This can be interpreted
as the codebook is first partitioned into $2^{nFR_{ji}}$ message
bins and then each bin is divided into $2^{nFR^1_{ji}}$ sub-bins
which we refer to the first layer of sub-bins. Each sub-bin in the
first layer is further divided into $2^{nFR^2_{ji}}$ sub-bins which
comprise the second layer. Such partition is repeated until the
$(N-1)^{th}$ layer. Each sub-bin in the last layer contains
$2^{nFR^{\dagger}_{ji}}$ codewords. Hence,
$w_{ji},b^1_{ji},\ldots,b^{N-1}_{ji}$ represent the message bin and
the sub-bin indexes of the $k^{th}, \forall k=1,\cdots,N-1$ layer
respectively.

Now, to send a message $w_{ji}$, Transmitter $i$ looks into the
message bin $w_{ji}$ and randomly selects a sub-bin $b^1_{ji}$ in
the first layer, sub-bin $b^2_{ji}$  in the second layer and so on
according to the uniform distribution. In the sub-bin of the last
layer a codeword $b^{\dagger}_{ji}$ is chosen uniformly over
$\{1,\dots,2^{nFR^{\dagger}_{ji}}\}$. Here, it obtains a codeword
$\xX^n(w_{ji},b^1_{ji},\cdots,b^{N-1}_{ji},b^{\dagger}_{ji})=[\xX_{ji}(1),\cdots,\xX_{ji}(n)]$.
For each time slot $t \in \{1,\ldots,n\}$, Transmitter $i$ employs
the interference alignment scheme mapping $\xX_{ji}(t)$ to
$\mathbf{V}_{ji}(t)\xX_{ji}(t)$ where $\mathbf{V}_{ji}$ is the $F
\times m_i$ matrix. At last, Transmitter $i$ sends signal
$\barxX_i(t)=\sum_{j=1}^N\mathbf{V}_{ji}(t)\xX_{ji}(t)$ into the
channel. Note that the pre-coding matrices $\mathbf{V}_{ji}(t)$ are
chosen as given in \cite{Cadambe_Jafar_X}.

Without loss of generality, we assume
\begin{eqnarray*}
I(\xX_{(\Jj-1)\times \Ii};\barxY_1|\xX_{1 \times \Ii}) <
I(\xX_{(\Jj-2)\times \Ii};\barxY_2|\xX_{2 \times \Ii})\\ < \cdots <
I(\xX_{(\Jj-N)\times \Ii};\barxY_N|\xX_{N \times \Ii})
\end{eqnarray*}
where $\xX_{(\Jj-j)\times \Ii}=\{\xX_{ri}: \forall (r,i) \in
(\Jj-j)\times \Ii\}$ and $\xX_{r \times \Ii}=\{\xX_{ri}: \forall i
\in \Ii\}$. We choose rates $R_{ji},
R^1_{ji},\cdots,R^{N-1}_{ji},R^{\dagger}_{ji}$ as follows
\begin{eqnarray*}
&R_{ji} =
\frac{1}{F}I(\xX_{ji};\barxY_j)\notag-\frac{1}{F}\frac{1}{M(N-1)}I(\xX_{(\Jj-N)\times
\Ii};\barxY_N|\xX_{N\times \Ii})\\ &R^1_{ji} =
\frac{1}{F}\frac{1}{M(N-1)}[I(\xX_{(\Jj-N)\times
\Ii};\barxY_N|\xX_{N\times \Ii}) \notag
\\&-I(\xX_{(\Jj-N+1)\times
\Ii};\barxY_{N-1}|\xX_{(N-1)\times \Ii})] \notag\\
&\vdots \notag
\end{eqnarray*}
\begin{eqnarray}
&R^k_{ji} = \frac{1}{F}\frac{1}{M(N-1)} \times
\notag\\
&[I(\xX_{(\Jj-N+k-1)\times \Ii};\barxY_{N-k+1}|\xX_{(N-k+1)\times
\Ii}) \notag
\\&-I(\xX_{(\Jj-N+k)\times \Ii};\barxY_{N-k}|\xX_{(N-k)\times \Ii})]\label{rk}\\
&\vdots \notag\\
&R^{N-1}_{ji}  = \frac{1}{F}\frac{1}{M(N-1)}[I(\xX_{(\Jj-2)\times
\Ii};\barxY_2|\xX_{2\times \Ii})\notag
\\&-I(\xX_{(\Jj-1)\times
\Ii};\barxY_{1}|\xX_{1\times \Ii})]\\
&R^{\dagger}_{ji} = \frac{1}{F}\frac{1}{M(N-1)}I(\xX_{(\Jj-1)\times
\Ii};\barxY_{1}|\xX_{1\times \Ii}))-\epsilon \label{rdagger}
\end{eqnarray}
Note that
$R_{ji}+R^1_{ji}+\cdots+R^{\dagger}_{ji}=\frac{1}{F}I(\xX_{ji};\barxY_j)-\epsilon$.
Next, we will show this scheme satisfies both the reliability
requirement and the secrecy constraint.

Since
$R_{ji}+R^1_{ji}+\cdots+R^{\dagger}_{ji}=\frac{1}{F}I(\xX_{ji};\barxY_j)-\epsilon
< \frac{1}{F}I(\xX_{ji};\barxY_j)$, each user can decode its desired
streams reliably.

To ensure the secrecy constraint 1, we need to show
\begin{eqnarray*}
H(\mathbf{W}_{(\Jj-j)\times \Ii}|\barxY_j^n) \geq \sum_{(r,i)\in
(\mathcal{J}-j) \times \mathcal{I}}R_{ri}-\epsilon \\
\Jj=\{1,\ldots,N\},~\Ii=\{1,\ldots,M\}
\end{eqnarray*}
We consider the following equivocation lower bound
\begin{equation}
H(\mathbf{W}_{(\Jj-j) \times \Ii}|\barxY_j^n) \geq
H(\mathbf{W}_{(\Jj-j) \times \Ii}|\barxY_j^n,\xx_{j\times
\Ii}^n)\label{lowerbound}
\end{equation}
where the inequality is due to the fact that conditioning reduces
entropy.
\begin{eqnarray}
&&H(\mathbf{W}_{(\Jj-j) \times \Ii}|\barxY_j^n,\xx_{j\times \Ii}^n)\notag\\
&=& H(\mathbf{W}_{(\Jj-j) \times \Ii},\barxY_j^n|\xx_{j\times
\Ii}^n)-H(\barxY_j^n|\xx_{j\times \Ii}^n)\\
&\geq & H(\mathbf{W}_{(\Jj-j) \times \Ii},\barxY_j^n|\xx_{j\times
\Ii}^n,\mathcal{B}^j_{(\Jj-j) \times \Ii})-H(\barxY_j^n|\xx_{j\times
\Ii}^n)\notag\\ \label{step1}
\end{eqnarray}
where $\mathcal{B}^j_{(\Jj-j) \times \Ii}=\{\mathbf{B}^1_{(\Jj-j)
\times \Ii},\mathbf{B}^2_{(\Jj-j) \times
\Ii},\cdots,\mathbf{B}^{N-j}_{(\Jj-j) \times \Ii}\}$ and
$\mathbf{B}^k_{(\Jj-j) \times \Ii}=\{B^k_{ri}:\forall (r,i) \in
(\Jj-j) \times \Ii\}, \forall k=1,\cdots,N-1$ denotes the set of all
the sub-bin indexes of the $k^{th}$ layer for all codewords
$\xx^n_{(\Jj-j) \times \Ii}$. $B^k_{ri}$ denotes the sub-bin index
in the $k^{th}$ layer for codeword $\xx^n_{ji}$ and is uniformly
distributed over $\{1,\ldots,2^{nFR^k_{ji}}\}$. Then, the first term
of \eqref{step1} can be written as
\begin{align}
&~~~~H(\mathbf{W}_{(\Jj-j) \times \Ii},\barxY_j^n|\xx_{j\times
\Ii}^n,\mathcal{B}^j_{(\Jj-j) \times
\Ii})\notag\\&=H(\mathbf{W}_{(\Jj-j) \times
\Ii},\barxY_j^n,\xx^n_{(\Jj-j) \times \Ii}|\xx_{j\times
\Ii}^n,\mathcal{B}^j_{(\Jj-j) \times
\Ii})\notag\\&~~~-H(\xx^n_{(\Jj-j) \times \Ii}|\mathbf{W}_{(\Jj-j)
\times \Ii},\barxY_j^n,\xx_{j\times \Ii}^n,\mathcal{B}^j_{(\Jj-j)
\times
\Ii})\notag\\
&=H(\mathbf{W}_{(\Jj-j) \times \Ii},\xx^n_{(\Jj-j) \times
\Ii}|\xx_{j\times \Ii}^n,\mathcal{B}^j_{(\Jj-j) \times
\Ii})\notag\\&~~~+H(\barxY_j^n|\mathbf{W}_{(\Jj-j) \times
\Ii},\xx^n_{(\Jj-j) \times \Ii},\xx_{j\times
\Ii}^n,\mathcal{B}^j_{(\Jj-j) \times
\Ii})\notag\\
&~~~-H(\xx^n_{(\Jj-j) \times \Ii}|\mathbf{W}_{(\Jj-j) \times
\Ii},\barxY_j^n,\xx_{j\times \Ii}^n,\mathcal{B}^j_{(\Jj-j) \times
\Ii})\notag\\
&=H(\mathbf{W}_{(\Jj-j) \times \Ii},\xx^n_{(\Jj-j) \times
\Ii}|\mathcal{B}^j_{(\Jj-j) \times
\Ii})\notag\\&~~~+H(\barxY_j^n|\xx^n_{(\Jj-j) \times
\Ii},\xx_{j\times
\Ii}^n)\notag\\
&~~~-H(\xx^n_{(\Jj-j) \times \Ii}|\mathbf{W}_{(\Jj-j) \times
\Ii},\barxY_j^n,\xx_{j\times \Ii}^n,\mathcal{B}^j_{(\Jj-j) \times
\Ii})\label{step2}
\end{align}
where
\begin{eqnarray*}
&H(\mathbf{W}_{(\Jj-j) \times \Ii},\xx^n_{(\Jj-j) \times
\Ii}|\xx_{j\times \Ii}^n,\mathcal{B}^j_{(\Jj-j) \times
\Ii})\\&=H(\mathbf{W}_{(\Jj-j) \times \Ii},\xx^n_{(\Jj-j) \times
\Ii}|\mathcal{B}^j_{(\Jj-j) \times \Ii})~~~~~~~~~~
\end{eqnarray*}
since $\mathbf{W}_{(\Jj-j) \times \Ii}$, $\xx^n_{(\Jj-j) \times
\Ii}$ are independent of $\xx_{j\times \Ii}^n$, and
\begin{eqnarray*}
&H(\barxY_j^n|\mathbf{W}_{(\Jj-j) \times \Ii},\xx^n_{(\Jj-j) \times
\Ii},\xx_{j\times \Ii}^n,\mathcal{B}^j_{(\Jj-j) \times
\Ii})\\&=H(\barxY_j^n|\xx^n_{(\Jj-j) \times \Ii},\xx_{j\times
\Ii}^n)~~~~~~~~~~~~~~~~~~~~~~~~~~~~~~~
\end{eqnarray*}
due to the Markov chain
\begin{equation*}
(\mathbf{W}_{(\Jj-j) \times \Ii},\mathcal{B}^j_{(\Jj-j) \times \Ii})
\to (\xx^n_{(\Jj-j) \times \Ii},\xx_{j\times \Ii}^n) \to \barxY_j^n
\end{equation*}
Hence, from \eqref{lowerbound}, \eqref{step1}, \eqref{step2}, we
obtain
\begin{eqnarray}
&&H(\mathbf{W}_{(\Jj-j) \times \Ii}|\barxY_j^n) \notag\\&&\geq
H(\mathbf{W}_{(\Jj-j) \times \Ii},\xx^n_{(\Jj-j) \times
\Ii}|\mathcal{B}^j_{(\Jj-j) \times
\Ii})\notag\\&&+H(\barxY_j^n|\xx^n_{(\Jj-j) \times \Ii},\xx_{j\times
\Ii}^n)\notag -H(\barxY_j^n|\xx_{j\times
\Ii}^n)\\&&-H(\xx^n_{(\Jj-j) \times \Ii}|\mathbf{W}_{(\Jj-j) \times
\Ii},\barxY_j^n,\xx_{j\times \Ii}^n,\mathcal{B}^j_{(\Jj-j) \times
\Ii}) \notag\\&& \geq H(\xx^n_{(\Jj-j) \times
\Ii}|\mathcal{B}^j_{(\Jj-j) \times \Ii})-I(\xx^n_{(\Jj-j)\times
\Ii},\barxY_j^n|\xx_{j\times \Ii}^n)\notag\\&&-H(\xx^n_{(\Jj-j)
\times \Ii}|\mathbf{W}_{(\Jj-j) \times \Ii},\barxY_j^n,\xx_{j\times
\Ii}^n),\mathcal{B}^j_{(\Jj-j) \times \Ii})\label{bound}
\end{eqnarray}
We now bound each term in $\eqref{bound}$. Consider the first term.
Note that given the first to $(N-j)^{th}$ layers' sub-bin indexes,
$\xx^n_{(\Jj-j) \times \Ii}$ has $2^{nF\sum_{(r,i)\in
\{\Jj-j\}\times
\Ii}(R^{N-j+1}_{ri}+\cdots+R^{N-1}_{ri}+R^{\dagger}_{ri})}$ possible
values with equal probability. Hence
\begin{eqnarray}
&&~~~H(\xx^n_{(\Jj-j) \times \Ii}|\mathcal{B}^j_{(\Jj-j) \times
\Ii})\notag\\&&= nF\sum_{(r,i)\in(\Jj-j)\times
\Ii}(R_{ri}+R^{N-j+1}_{ri}+\cdots+R^{N-1}_{ri}+R^{\dagger}_{ri}) \notag\\
&&=nF\sum_{(r,i)\in(\Jj-j)\times \Ii}R_{ri}+ n I(\xX_{(\Jj-j)\times
\Ii};\barxY_j|\xX_{j \times \Ii})-\epsilon_1\notag\\ \label{b1}
\end{eqnarray}
where the last step follows from $\eqref{rk}$ and \eqref{rdagger}.
$\epsilon_1 \to 0$ as $n \to \infty$. Second, we can bound
\begin{equation}
I(\xx^n_{(\Jj-j) \times \Ii},\barxY_j^n|\xx_{j\times \Ii}^n)\leq
nI(\xx_{(\Jj-j) \times \Ii},\barxY_j|\xx_{j\times
\Ii})+n\epsilon_2\label{b3}
\end{equation}
where $\epsilon_2 \to 0$ as $n \to \infty$. Finally, the third term
can be bounded as follows
\begin{equation} H(\xx^n_{(\Jj-j) \times
\Ii}|\mathbf{W}_{(\Jj-j) \times \Ii},\barxY_j^n,\xx_{j\times
\Ii}^n,\mathcal{B}^j_{(\Jj-j) \times \Ii})\leq n\epsilon_3\label{b2}
\end{equation}
where $\epsilon_3 \to 0$ as $n \to \infty$. This is because Receiver
$j$ can decode the codeword $\xx^n_{(\Jj-j) \times \Ii}$ given the
message, the first to $(N-j)^{th}$ layers' sub-bin indexes and the
observation $\barxY_j^n$. Then, Fano's inequality implies
\eqref{b2}.

From $\eqref{b1}$, $\eqref{b3}$ and $\eqref{b2}$ , we can write
\eqref{bound} as
\begin{eqnarray*}
&&\frac{1}{nF}H(\mathbf{W}_{(\Jj-j) \times \Ii}|\barxY_j^n)\\
&&\geq \sum_{(r,i)\in(\Jj-j)\times \Ii}R_{ri}+
\frac{1}{F}I(\xX_{(\Jj-j)\times \Ii};\barxY_j|\xX_{j \times
\Ii})\\&&-\frac{1}{F}I(\xX_{(\Jj-j)\times \Ii};\barxY_j|\xX_{j
\times \Ii})-\epsilon_1-\epsilon_2-\epsilon_3
\end{eqnarray*}
Hence, security condition is satisfied at Receiver $j$. Therefore,
the following secrecy rate is achievable:
\begin{equation*}
R_{ji}=\frac{1}{F}I(\xX_{ji};\barxY_j)-\frac{1}{F}\frac{1}{M(N-1)}I(\xX_{(\Jj-N)\times
\Ii};\barxY_N|\xX_{N\times \Ii})
\end{equation*}
From \cite{Cadambe_Jafar_X}, we have
\begin{equation*}
I(\xX_{ji};\barxY_j)=\eta \log(P)+o(\log(P))
\end{equation*}
where $\eta=(m+1)^{\Gamma}$ when $i=1$ and $\eta=m^{\Gamma}$ when
$i=2,3,\ldots,M$. At each receiver, the interference vectors from
Transmitter $2,3\ldots,M$ are aligned perfectly with the
interference from Transmitter 1. Then, we have
\begin{equation*}
I(\xX_{(\Jj-N)\times \Ii};\barxY_N|\xX_{N\times
\Ii})=(N-1)(m+1)^{\Gamma}\log(P)+o(\log(P))
\end{equation*}
Hence,
\begin{eqnarray*}
R_{ji}&=&\frac{1}{F}I(\xX_{ji};\barxY_j)-\frac{1}{F}\frac{1}{M(N-1)}I(\xX_{(\Jj-N)\times
\Ii};\barxY_N|\xX_{N\times \Ii})\\
&=&\frac{1}{F}(m+1)^{\Gamma}(1-\frac{1}{M})\log(P)+o(\log(P))~~i=1
\end{eqnarray*}
and
\begin{equation*}
R_{ji}=
\frac{1}{F}(m^{\Gamma}-\frac{(m+1)^{\Gamma}}{M})\log(P)+o(\log(P))~~i=2,\ldots,M
\end{equation*}
As $m \to \infty$, we have
\begin{eqnarray*}
R_{ji}=\frac{M-1}{M(M+N-1)}\log(P)+o(\log(P))\\ \forall (j,i)\in
\{1,\ldots,N\} \times \{1,2,\ldots,M\}
\end{eqnarray*}
As a result, each message can achieve
$\eta_{ji}=\frac{M-1}{M(M+N-1)}$ secure degrees of freedom.
Therefore, for a total of $MN$ messages, we can achieve a total of
$\frac{N(M-1)}{M+N-1}$ secure degrees of freedom. The proof is
complete.
\end{proof}

\end{document}